\documentclass[conference,letterpaper]{IEEEtran}

\usepackage{times}
\usepackage[T1]{fontenc}
\usepackage{microtype}
\usepackage[utf8]{inputenc}
\usepackage[cmex10]{amsmath}
\usepackage{amsthm,amsfonts,amsbsy,amssymb,mleftright}
\usepackage{newtxmath,mathrsfs,bbold,bbm,dsfont}
\usepackage{enumitem}
\usepackage{graphicx,float}
\usepackage[caption=false]{subfig}
\usepackage{array,longtable,booktabs,makecell}
\usepackage{tikz,xcolor}
\usepackage{verbatim,algpseudocode,listings}
\usepackage{csquotes}
\usepackage{hyperref}
\usepackage{lipsum}
\usepackage{url}
\usepackage{ifthen}
\usepackage{cite}

\newtheorem{theorem}{Theorem}
\newtheorem{corollary}[theorem]{Corollary}

\newtheorem{lemma}[theorem]{Lemma}
\newtheorem{proposition}[theorem]{Proposition}

\theoremstyle{definition}

\newcommand{\1}{I}

\newcommand{\aff}{\mathrm{aff}}
\newcommand{\itr}{\mathrm{int}}

\newcommand{\tr}{\mathrm{Tr}}

\newcommand{\abb}[1]{{\textnormal{#1}}} 
\newcommand{\en}[1]{{\mathscr{#1}}} 
\newcommand{\ch}[1]{{\mathcal{#1}}} 
\newcommand{\g}[1]{{\widehat{#1}}} 
\newcommand{\s}[1]{{\mathsf{#1}}} 
\newcommand{\sw}[1]{{\widetilde{#1}}} 
\newcommand{\spa}[1]{{\mathds{#1}}} 

\newcommand{\bra}[1]{\langle#1\rvert}
\newcommand{\ket}[1]{\lvert#1\rangle}

\newcommand{\op}[2]{\ket{#1}\!\bra{#2}}

\newcommand{\mclose}{\mathclose{}}
\newcommand{\fleft}{\mathopen{}\left}
\newcommand{\fright}{\aftergroup\mclose\right}

\definecolor{darkblue}{rgb}{0,0,0.6}
\definecolor{darkgreen}{rgb}{0,0.45,0.1}
\definecolor{darkred}{rgb}{0.5,0,0}
\hypersetup{
	colorlinks = true,
	citecolor = darkgreen, 
	linkcolor = darkblue, 
	urlcolor  = darkblue, 
	filecolor = darkred 
}

\allowdisplaybreaks
\begin{document}


\title{Converse bounds for quantum hypothesis exclusion: A divergence-radius approach}

\author{
\IEEEauthorblockN{
Kaiyuan Ji,\IEEEauthorrefmark{1} Hemant K. Mishra,\IEEEauthorrefmark{2}\IEEEauthorrefmark{1} Milán Mosonyi,\IEEEauthorrefmark{3}\IEEEauthorrefmark{4} Mark M. Wilde\IEEEauthorrefmark{1}
}

\IEEEauthorblockA{
\IEEEauthorrefmark{1} School of Electrical and Computer Engineering, Cornell University, Ithaca, New York 14850, USA \\
Email: \{kj264,hkm37,wilde\}@cornell.edu \\
\IEEEauthorrefmark{2} Department of Mathematics and Computing, Indian Institute of Technology \\ (Indian School of Mines) Dhanbad, Jharkhand 826004, India \\
\IEEEauthorrefmark{3} HUN-REN Alfréd Rényi Institute of Mathematics, Reáltanoda Street 13-15, H-1053, Budapest, Hungary \\
\IEEEauthorrefmark{4} Department of Analysis and Operations Research, Institute of Mathematics, \\ Budapest University of Technology and Economics, Műegyetem rkp.\ 3., H-1111 Budapest, Hungary \\
}
}

\maketitle


\begin{abstract}
Hypothesis exclusion is an information-theoretic task in which an experimenter aims at ruling out a false hypothesis from a finite set of known candidates, and an error occurs if and only if the hypothesis being ruled out is the ground truth.  For the tasks of quantum state exclusion and quantum channel exclusion --- where hypotheses are represented by quantum states and quantum channels, respectively --- efficiently computable upper bounds on the asymptotic error exponents were established in a recent work of the current authors~[Ji~\textit{et al.}, arXiv:2407.13728 (2024)], where the derivation was based on nonasymptotic analysis.  In this companion paper of our previous work, we provide alternative proofs for the same upper bounds on the asymptotic error exponents of quantum state and channel exclusion, but using a conceptually different approach from the one adopted in the previous work.  Specifically, we apply strong converse results for asymmetric binary hypothesis testing to distinguishing an arbitrary ``dummy'' hypothesis from each of the concerned candidates.  This leads to the desired upper bounds in terms of divergence radii via a geometrically inspired argument.
\end{abstract}



\section{Introduction}
\label{sec:introduction}

Hypothesis testing is a fundamental task in quantum information theory~\cite{wilde2017QuantumInformationTheory,hayashi2017QuantumInformationTheory,watrous2018TheoryQuantumInformation,khatri2024PrinciplesQuantumCommunication}.  Usually taking the form of quantum state or channel discrimination, it involves an experimenter trying to determine the true hypothesis --- the identity of a given state or channel --- out of a finite number of known possibilities.  In this paper, we consider a variation of hypothesis testing, termed \emph{hypothesis exclusion}, in which one is not expected to figure out the true identity completely but is merely asked to choose one false hypothesis to rule out.  For instance, in quantum state exclusion, the experimenter is given a state randomly selected from a tuple $(\rho_1,\rho_2,\dots,\rho_r)$ under a prior probability distribution $(p_1,p_2,\dots,p_r)$ and is asked to rule out one of the $r$ possibilities.  The error probability, i.e., the probability of ``falsely'' ruling out the true state, is given by $\sum_{x\in[r]}p_x\tr[\Lambda_x\rho_x]$, where $(\Lambda_1,\Lambda_2,\dots,\Lambda_r)$ is the positive operator-valued measure (POVM) that leads to the experimenter's decision on which state to rule out.

Quantum state exclusion has found its significance early in quantum foundations, such as in the studies of compatibility of quantum-state assignment~\cite{brun2002HowMuchState,caves2002ConditionsCompatibilityQuantumstate} and ontological interpretations of quantum states~\cite{pusey2012RealityQuantumState,leifer2014QuantumStateReal,barrett2014NoPsepistemicModel}.  The study of the information-theoretic limit of state exclusion was initiated recently in Ref.~\cite{mishra2024OptimalErrorExponents}, where the exact error exponent of classical state exclusion was determined and bounds were derived more generally for quantum state exclusion.

In a recent work of the current authors~\cite{ji2024BarycentricBoundsError}, further advancements were made on bounding the information-theoretic limit of quantum state exclusion and that of the more complicated task of quantum channel exclusion.  Specifically, we showed that the error exponent of quantum state exclusion is bounded from above by the multivariate log-Euclidean Chernoff divergence; we also derived an efficiently computable upper bound on the error exponent of quantum channel exclusion in terms of a barycentric Chernoff divergence, which we further showed to be achievable for classical channel exclusion.  These bounds are the tightest efficiently computable upper bounds known so far. They were established by expressing the one-shot error probability for quantum state exclusion in terms of an optimization of the hypothesis-testing divergence and analyzing the asymptotics of the latter.  This manuscript is a companion paper of our previous work~\cite{ji2024BarycentricBoundsError} and provides additional insights into the upper bounds on the error exponents of quantum state and channel exclusion established therein.  To do so, we present alternative proofs of the same upper bounds but using a method conceptually different from the one originally adopted in Ref.~\cite{ji2024BarycentricBoundsError}.  Our new method is inspired by the ``divergence sphere'' approach in Ref.~\cite[Exercise~3.57]{hayashi2017QuantumInformationTheory} and adapts it to a multi-hypotheses setting related to state exclusion, thus termed a ``divergence radius'' approach.  Specifically, we employ a strong converse result of asymmetric binary hypothesis testing and apply it to distinguishing an arbitrary ``dummy'' hypothesis from each of the concerned hypotheses; the normalization of any complete exclusion strategy gives rise to an upper bound on the error exponent of hypothesis exclusion in terms of a divergence radius.  This new approach offers an immediate and geometrically relevant explanation for why divergence radii serve as natural converse bounds for hypothesis exclusion tasks.

\section{Preliminaries}
\label{sec:preliminaries}

\subsection{Notation}
\label{sec:notation}

Let $[r]\equiv\{1,2,\dots,r\}$ denote the set of $r$ smallest distinct positive integers.  Throughout the paper, we use $\gamma_{[r]}\equiv(\gamma_1,\gamma_2,\dots,\gamma_r)$ to denote a tuple of entities with indices from $[r]$, regardless of the nature of these entities.  We also use $\s{S}^{[r]}\equiv\{\gamma_{[r]}\colon\gamma_x\in\s{S}\;\forall x\in[r]\}$ to denote the set of tuples each of whose entities belongs to a set $\s{S}$.

Let $\spa{H}_A$ denote the (finite-dimensional) Hilbert space associated with a quantum system $A$.  Let $\spa{B}_A$ denote the space of bounded operators acting on $\spa{H}_A$.  Let $\s{Herm}_A$ and $\s{PSD}_A$ denote the set of Hermitian operators and the set of positive semidefinite operators in $\spa{B}_A$, respectively.  For a Hermitian operator $\gamma\in\s{Herm}_A$, let $\gamma^0\in\s{PSD}_A$ denote the projector onto the support of $\gamma$.  For a positive semidefinite operator $\sigma\in\s{PSD}_A$, the negative power $\sigma^{-\alpha}$ for $\alpha\in[0,+\infty)$ and the logarithm $\ln\sigma$ are both taken on the support of $\sigma$.

Let $\s{D}_A$ denote the set of quantum states (i.e., unit-trace positive semidefinite operators) of a system $A$.  Let $\aff(\s{D}_A)$ denote the set of unit-trace Hermitian operators in $\spa{B}_A$, which is also the affine hull of $\s{D}_A$.  Let $\s{CP}_{A\to B}$ and $\s{C}_{A\to B}$ denote the set of completely positive (CP) maps and the set of quantum channels (i.e., completely positive trace-preserving maps) from $\spa{B}_A$ to $\spa{B}_B$, respectively.  Let $J_\ch{M}\equiv\sum_{i,j\in[d_A]}\op{i}{j}_R\otimes\ch{M}_{A\to B}[\op{i}{j}_A]$ denote the Choi operator of a CP map $\ch{M}\in\s{CP}_{A\to B}$, where $R$ is a system such that $d_R=d_A$.  Let $\s{M}_{A,r}$ denote the set of POVMs (i.e., tuples of positive semidefinite operators summing to the identity operator) on a system $A$ with $r$ possible outcomes.  Let $\s{P}_r$ denote the set of probability distributions over $[r]$.  Let $\itr(\s{P}_r)$ denote the set of probability distributions whose support is $[r]$, which is also the interior of $\s{P}_r$.

\subsection{Divergence measures}
\label{sec:divergence}

Let $\rho\in\s{D}_A$ be a state, and let $\sigma\in\s{PSD}_A$ be a positive semidefinite operator.  For $\alpha\in(1,+\infty)$, the \emph{sandwiched Rényi divergence} is defined as~\cite{muller-lennert2013QuantumRenyiEntropies,wilde2014StrongConverseClassical}
\begin{align}
	\label{eq:sandwiched}
	\sw{D}_\alpha\fleft(\rho\middle\|\sigma\fright)&\coloneq\begin{cases}
		\frac{1}{\alpha-1}\ln\left\lVert\sigma^\frac{1-\alpha}{2\alpha}\rho\sigma^\frac{1-\alpha}{2\alpha}\right\rVert_\alpha^\alpha&\text{if }\rho^0\leq\sigma^0, \\
		+\infty&\text{otherwise},
	\end{cases}
\end{align}
where $\lVert\gamma\rVert_\alpha\coloneq(\tr[(\gamma^\dagger\gamma)^\frac{\alpha}{2}])^\frac{1}{\alpha}$ denotes the $\alpha$-norm of an operator $\gamma\in\spa{B}_A$ for $\alpha\in[1,+\infty)$.  As shown in Refs.~\cite{muller-lennert2013QuantumRenyiEntropies,wilde2014StrongConverseClassical}, the limit of the sandwiched Rényi divergence as $\alpha\searrow1$ is given by the \emph{Umegaki divergence}~\cite{umegaki1962ConditionalExpectationOperator}, which is defined as
\begin{align}
	D\fleft(\rho\middle\|\sigma\fright)&\coloneq\begin{cases}
		\tr\fleft[\rho\left(\ln\rho-\ln\sigma\right)\fright]&\text{if }\rho^0\leq\sigma^0, \\
		+\infty&\text{otherwise},
	\end{cases} \\
	&=\lim_{\alpha\searrow1}\sw{D}_\alpha\fleft(\rho\middle\|\sigma\fright).
\end{align}
The \emph{geometric Rényi divergence} is defined for $\alpha\in(1,2]$ as~\cite{matsumoto2018NewQuantumVersion}
\begin{align}
	\g{D}_\alpha\fleft(\rho\middle\|\sigma\fright)&\coloneq\begin{cases}
		\frac{1}{\alpha-1}\ln\tr\fleft[\sigma\left(\sigma^{-\frac{1}{2}}\rho\sigma^{-\frac{1}{2}}\right)^\alpha\fright]&\text{if }\rho^0\leq\sigma^0, \\
		+\infty&\text{otherwise}.
	\end{cases}
\end{align}
As shown in Ref.~\cite[Proposition~79]{katariya2021GeometricDistinguishabilityMeasures}, the limit of the geometric Rényi divergence as $\alpha\searrow1$ is given by the \emph{Belavkin--Staszewski divergence}~\cite{belavkin1982$C^$algebraicGeneralizationRelative}, which is defined as
\begin{align}
	\g{D}\fleft(\rho\middle\|\sigma\fright)&\coloneq\begin{cases}
		\tr\fleft[\rho\ln\left(\rho^\frac{1}{2}\sigma^{-1}\rho^\frac{1}{2}\right)\fright]&\text{if }\rho^0\leq\sigma^0, \\
		+\infty&\text{otherwise},
	\end{cases} \\
	&=\lim_{\alpha\searrow1}\g{D}_\alpha\fleft(\rho\middle\|\sigma\fright).
\end{align}

For $\alpha\in(1,2]$, the \emph{geometric Rényi channel divergence} is defined as
\begin{align}
	&\g{D}_\alpha\fleft(\ch{N}\middle\|\ch{M}\fright)\coloneq\sup_{\rho\in\s{D}_{RA}}\g{D}_\alpha\fleft(\ch{N}_{A\to B}\fleft[\rho_{RA}\fright]\middle\|\ch{M}_{A\to B}\fleft[\rho_{RA}\fright]\fright) \\
	&\quad=\begin{cases}
		\frac{1}{\alpha-1}\ln\left\lVert\tr_B\fleft[J_\ch{M}^\frac{1}{2}\left(J_\ch{M}^{-\frac{1}{2}}J_\ch{N}J_\ch{M}^{-\frac{1}{2}}\right)^\alpha J_\ch{M}^\frac{1}{2}\fright]\right\rVert_\infty&\text{if }J_\ch{N}^0\leq J_\ch{M}^0, \\
		+\infty&\text{otherwise}.
	\end{cases} \label{eq:geometric-channel-choi}
\end{align}
As shown in Ref.~\cite[Lemma~35]{ding2023BoundingForwardClassical}, the limit of the geometric Rényi channel divergence as $\alpha\searrow1$ is given by the \emph{Belavkin--Staszewski channel divergence}:
\begin{align}
	&\g{D}\fleft(\ch{N}\middle\|\ch{M}\fright)\coloneq\sup_{\rho\in\s{D}_{RA}}\g{D}\fleft(\ch{N}_{A\to B}\fleft[\rho_{RA}\fright]\middle\|\ch{M}_{A\to B}\fleft[\rho_{RA}\fright]\fright) \\
	&\quad=\begin{cases}
		\left\lVert\tr_B\fleft[J_\ch{N}^\frac{1}{2}\ln\left(J_\ch{N}^\frac{1}{2}J_\ch{M}^{-1}J_\ch{N}^\frac{1}{2}\right)J_\ch{N}^\frac{1}{2}\fright]\right\rVert_\infty&\text{if }J_\ch{N}^0\leq J_\ch{M}^0, \\
		+\infty&\text{otherwise},
	\end{cases} \label{eq:belavkin-channel-choi}\\
	&\quad=\lim_{\alpha\searrow1}\g{D}_\alpha\fleft(\ch{N}\middle\|\ch{M}\fright).
\end{align}
The closed-form expressions in \eqref{eq:geometric-channel-choi} and \eqref{eq:belavkin-channel-choi} were established in Ref.~\cite[Theorem~3.2]{fang2021GeometricRenyiDivergence}.

Let $\rho_{[r]}\in\s{PSD}_A^{[r]}$ be a tuple of positive semidefinite operators.  The \emph{multivariate log-Euclidean Chernoff divergence}~\cite{mishra2024OptimalErrorExponents} (also see Ref.~\cite{mosonyi2024GeometricRelativeEntropies}) is defined as
\begin{align}
	C^\flat\fleft(\rho_{[r]}\fright)&\coloneq\sup_{s_{[r]}\in\s{P}_r}\lim_{\varepsilon\searrow0}-\ln\tr\fleft[\exp\left(\sum_{x\in[r]}s_x\ln\left(\rho_x+\varepsilon\1\right)\right)\fright]  \\
	&=\sup_{s_{[r]}\in\s{P}_r}\inf_{\tau\in\s{D}_A}\sum_{x\in[r]}s_xD\fleft(\tau\middle\|\rho_x\fright). \label{eq:euclidean}
\end{align}

\subsection{Extended sandwiched Rényi divergence}
\label{sec:extended}

As proposed in Ref.~\cite{wang2020AlogarithmicNegativity}, the definition of the sandwiched Rényi divergence can be generalized to an extended domain, with its first argument allowed to be a Hermitian (and not necessarily positive semidefinite) operator.

Let $\gamma\in\s{Herm}_A$ be a Hermitian operator with $\gamma\neq0$, and let $\sigma\in\s{PSD}_A$ be a positive semidefinite operator.  For $\alpha\in(1,+\infty)$, the \emph{extended sandwiched Rényi divergence} is defined as
\begin{align}
	\label{eq:sandwiched-extended}
	\sw{D}_\alpha\fleft(\gamma\middle\|\sigma\fright)&\coloneq\begin{cases}
		\frac{1}{\alpha-1}\ln\left\lVert\sigma^\frac{1-\alpha}{2\alpha}\gamma\sigma^\frac{1-\alpha}{2\alpha}\right\rVert_\alpha^\alpha&\text{if }\gamma^0\leq\sigma^0, \\
		+\infty&\text{otherwise}.
	\end{cases}
\end{align}
The extended definition of the sandwiched Rényi divergence has precisely the same formula as the original definition presented in \eqref{eq:sandwiched}, except for an enlarged domain of its first argument.  The extended sandwiched Rényi divergence has been shown to enjoy a variety of desirable properties, including the data-processing inequality~\cite[Lemma~2]{wang2020AlogarithmicNegativity} and additivity~\cite[Theorem~5.3]{ji2024BarycentricBoundsError}; see our companion paper~\cite{ji2024BarycentricBoundsError} for details and further discussions.

\section{Quantum state exclusion}
\label{sec:exclusion-state}

In the task of quantum state exclusion, an experimenter receives a system $A$ in an unknown state.  The source of the system is represented by an ensemble of states, $\en{E}\equiv(p_{[r]},\rho_{[r]})$ with $r\geq2$, $p_{[r]}\in\itr(\s{P}_r)$, and $\rho_{[r]}\in\s{D}_A^{[r]}$, and this indicates that for each $x\in[r]$, there is a prior probability $p_x\in(0,1)$ with which the state of the system is $\rho_x\in\s{D}_A$.  The experimenter's goal is to submit an index $x'\in[r]$ that \emph{differs} from the actual label of the state they received.

The most general strategy of the experimenter for state exclusion is represented by a POVM $\Lambda_{[r]}\in\s{M}_{A,r}$, which corresponds to performing a measurement on the system $A$ and submitting the measurement outcome.  An error occurs if and only if the outcome coincides with the actual label of the state.  Consequently, the \emph{(one-shot) error probability} of state exclusion for the ensemble $\en{E}$ is given by 
\begin{align}
	P_\abb{err}\fleft(\en{E}\fright)&\coloneq\inf_{\Lambda_{[r]}\in\s{M}_{A,r}}\sum_{x\in[r]}p_x\tr\fleft[\Lambda_x\rho_x\fright].
\end{align}
The \emph{(asymptotic) error exponents} of state exclusion for the ensemble $\en{E}$ are defined as
\begin{align}
	\underline{E}_\abb{err}\fleft(\en{E}\fright)&\coloneq\liminf_{n\to+\infty}\sup_{\Lambda_{[r]}^{(n)}\in\s{M}_{A^n,r}}-\frac{1}{n}\ln\left(\sum_{x\in[r]}p_x\tr\fleft[\Lambda_x^{(n)}\rho_x^{\otimes n}\fright]\right), \notag\\
	\overline{E}_\abb{err}\fleft(\en{E}\fright)&\coloneq\limsup_{n\to+\infty}\sup_{\Lambda_{[r]}^{(n)}\in\s{M}_{A^n,r}}-\frac{1}{n}\ln\left(\sum_{x\in[r]}p_x\tr\fleft[\Lambda_x^{(n)}\rho_x^{\otimes n}\fright]\right),
\end{align}
where $A^n$ denotes the system consisting of $n$ copies of $A$.

\subsection{Upper bound on the asymptotic error exponent}
\label{sec:asymptotic-state}

We provide an alternative proof for the log-Euclidean upper bound on the asymptotic error exponent $\underline{E}_\abb{err}(\en{E})$ of state exclusion, which was first established in our companion paper~\cite[Theorem~15]{ji2024BarycentricBoundsError}.  The proof presented here relates the task of state exclusion to that of asymmetric binary hypothesis testing, by utilizing the strong converse part of the quantum Stein's lemma~\cite{nagaoka2000StrongConverseSteins}.  It draws inspiration from Ref.~\cite[Exercise~3.57]{hayashi2017QuantumInformationTheory}.  See Fig.~\ref{fig} for an illustration of the proof idea.

\begin{lemma}[\!\!\cite{nagaoka2000StrongConverseSteins}]
\label{lem:converse-stein}
Let $\tau,\rho\in\s{D}_A$ be two states, and let $(\Lambda^{(n)})_n$ be a sequence of positive semidefinite operators such that $\Lambda^{(n)}\in\s{PSD}_{A^n}$ and $\Lambda^{(n)}\leq\1$ for every positive integer $n$.  If
\begin{align}
	\liminf_{n\to+\infty}-\frac{1}{n}\ln\tr\fleft[\Lambda^{(n)}\rho^{\otimes n}\fright]&>D\fleft(\tau\middle\|\rho\fright),
\end{align}
then
\begin{align}
	\limsup_{n\to+\infty}\tr\fleft[\Lambda^{(n)}\tau^{\otimes n}\fright]&=0.
\end{align}
\end{lemma}

\begin{theorem}
\label{thm:asymptotic-state}
Let $\rho_{[r]}\in\s{D}_A^{[r]}$ be a tuple of states.  Then
\begin{align}
	\underline{E}_\abb{err}\fleft(\en{E}\fright)&\leq C^\flat\fleft(\rho_{[r]}\fright).
\end{align}
\end{theorem}

\begin{proof}
Let $\tau\in\s{D}_A$ be a state, and let $(\Lambda_{[r]}^{(n)})_n$ be a sequence of POVMs such that $\Lambda_{[r]}^{(n)}\in\s{M}_{A^n,r}$ for every positive integer~$n$.  We assert that there exists $x_\star\in[r]$ such that
\begin{align}
	\label{pf:asymptotic-state-1}
	\liminf_{n\to+\infty}-\frac{1}{n}\ln\tr\fleft[\Lambda_{x_\star}^{(n)}\rho_{x_\star}^{\otimes n}\fright]&\leq D\fleft(\tau\middle\|\rho_{x_\star}\fright).
\end{align}
To see this, consider the opposite situation where
\begin{align}
	\liminf_{n\to+\infty}-\frac{1}{n}\ln\tr\fleft[\Lambda_x^{(n)}\rho_x^{\otimes n}\fright]&>D\fleft(\tau\middle\|\rho_x\fright)\quad\forall x\in[r].
\end{align}
By Lemma~\ref{lem:converse-stein}, this implies that
\begin{align}
	\limsup_{n\to+\infty}\tr\fleft[\Lambda_x^{(n)}\tau^{\otimes n}\fright]&=0\quad\forall x\in[r],
\end{align}
which contradicts the fact that $\sum_{x\in[r]}\tr[\Lambda_x^{(n)}\tau^{\otimes n}]=1$ for every positive integer $n$.  This shows by contradiction the existence of $x_\star\in[r]$ satisfying \eqref{pf:asymptotic-state-1}.  It follows from \eqref{pf:asymptotic-state-1} that
\begin{align}
	\min_{x\in[r]}\liminf_{n\to+\infty}-\frac{1}{n}\ln\tr\fleft[\Lambda_x^{(n)}\rho_x^{\otimes n}\fright]&\leq\max_{x\in[r]}D\fleft(\tau\middle\|\rho_x\fright).
\end{align}
Since this holds for every state $\tau$ and every sequence of POVMs $(\Lambda_{[r]}^{(n)})_n$, we infer that
\begin{align}
	\sup_{\mleft(\Lambda_{[r]}^{(n)}\mright)_n}\min_{x\in[r]}\liminf_{n\to+\infty}-\frac{1}{n}\ln\tr\fleft[\Lambda_x^{(n)}\rho_x^{\otimes n}\fright]&\leq\inf_{\tau\in\s{D}_A}\max_{x\in[r]}D\fleft(\tau\middle\|\rho_x\fright). \label{pf:asymptotic-state-2}
\end{align}
Then
\begin{align}
	\underline{E}_\abb{err}\fleft(\en{E}\fright)&=\liminf_{n\to+\infty}\sup_{\Lambda_{[r]}^{(n)}\in\s{M}_{A^n,r}}-\frac{1}{n}\ln\left(\sum_{x\in[r]}p_x\tr\fleft[\Lambda_x^{(n)}\rho_x^{\otimes n}\fright]\right) \label{pf:asymptotic-state-3}\\
	&=\sup_{\mleft(\Lambda_{[r]}^{(n)}\mright)_n}\liminf_{n\to+\infty}-\frac{1}{n}\ln\left(\sum_{x\in[r]}p_x\tr\fleft[\Lambda_x^{(n)}\rho_x^{\otimes n}\fright]\right) \label{pf:asymptotic-state-4}\\
	&=\sup_{\mleft(\Lambda_{[r]}^{(n)}\mright)_n}\liminf_{n\to+\infty}-\frac{1}{n}\ln\left(\max_{x\in[r]}\tr\fleft[\Lambda_x^{(n)}\rho_x^{\otimes n}\fright]\right) \\
	&\leq\sup_{\mleft(\Lambda_{[r]}^{(n)}\mright)_n}\min_{x\in[r]}\liminf_{n\to+\infty}-\frac{1}{n}\ln\left(\tr\fleft[\Lambda_x^{(n)}\rho_x^{\otimes n}\fright]\right) \\
	&\leq\inf_{\tau\in\s{D}_A}\max_{x\in[r]}D\fleft(\tau\middle\|\rho_x\fright) \label{pf:asymptotic-state-5}\\
	&=C^\flat\fleft(\rho_{[r]}\fright). \label{pf:asymptotic-state-6}
\end{align}
Here \eqref{pf:asymptotic-state-4} follows because the supremum on the right-hand side of \eqref{pf:asymptotic-state-3} can be replaced with a maximum; Eq.~\eqref{pf:asymptotic-state-5} follows from \eqref{pf:asymptotic-state-2}; Eq.~\eqref{pf:asymptotic-state-6} follows from Ref.~\cite[Proposition~A.1]{mosonyi2021DivergenceRadiiStrong} and \eqref{eq:euclidean} (also see Ref.~\cite[Eq.~(45)]{ji2024BarycentricBoundsError}).
\end{proof}

Theorem~\ref{thm:asymptotic-state} basically recovers the log-Euclidean upper bound on the error exponent using just a qualitative statement about the strong converse behaviour of asymmetric binary hypothesis testing.  However, the upper bound in Theorem~\ref{thm:asymptotic-state} is on $\underline{E}_\abb{err}(\en{E})$ instead of $\overline{E}_\abb{err}(\en{E})$, and in this sense there is a gap between the statement here and the original statement of Ref.~\cite[Theorem~15]{ji2024BarycentricBoundsError}.  In what follows, we fill this gap by providing a one-shot analysis corresponding to Theorem~\ref{thm:asymptotic-state}.

\subsection{Converse bound on the one-shot error probability}
\label{sec:oneshot-state}

The converse bound on the one-shot error probability of state exclusion in our companion paper~\cite[Proposition~12]{ji2024BarycentricBoundsError} can be recovered using an alternative method similar to the above, by making connections to the strong converse analysis of asymmetric binary hypothesis testing in the one-shot regime.  To do so, we first prove a generalization of Ref.~\cite[Eq.~(75)]{mosonyi2015QuantumHypothesisTesting} below, which may be of independent interest.

\begin{lemma}
\label{lem:converse-hoeffding}
Let $\tau\in\aff(\s{D}_A)$ be a unit-trace Hermitian operator, and let $\rho\in\s{D}_A$ be a state.  Let $\Lambda\in\s{PSD}_A$ be a positive semidefinite operator such that $\Lambda\leq\1$.  Then for all $\alpha\in(1,+\infty)$,
\begin{align}
	\left\lvert\tr\fleft[\Lambda\tau\fright]\right\rvert&\leq\left(\tr\fleft[\Lambda\rho\fright]\right)^\frac{\alpha-1}{\alpha}\exp\left(\frac{\alpha-1}{\alpha}\sw{D}_\alpha\fleft(\tau\middle\|\rho\fright)\right).
\end{align}
\end{lemma}

\begin{proof}
Under the following measurement channel:
\begin{align}
	\ch{M}\in\s{C}_{A\to X}&\colon\rho\mapsto\tr\fleft[\Lambda\rho\fright]\op{1}{1}+\tr\fleft[\left(\1-\Lambda\right)\rho\fright]\op{2}{2},
\end{align}
the data-processing inequality of the extended sandwiched Rényi divergence~\cite[Lemma~2]{wang2020AlogarithmicNegativity} gives that
\begin{align}
	\sw{D}_\alpha\fleft(\tau\middle\|\rho\fright)&\geq\sw{D}_\alpha\fleft(\ch{M}\fleft[\tau\fright]\middle\|\ch{M}\fleft[\rho\fright]\fright) \\
	&=\frac{1}{\alpha-1}\ln\left\lVert\left(\ch{M}\fleft[\rho\fright]\right)^\frac{1-\alpha}{2\alpha}\ch{M}\fleft[\tau\fright]\left(\ch{M}\fleft[\rho\fright]\right)^\frac{1-\alpha}{2\alpha}\right\rVert_\alpha^\alpha \\
	&=\frac{1}{\alpha-1}\ln\left(\left\lvert\tr\fleft[\Lambda\tau\fright]\tr\fleft[\Lambda\rho\fright]^\frac{1-\alpha}{\alpha}\right\rvert^\alpha\right. \notag\\
	&\quad\left.\mathop{+}\left\lvert\tr\fleft[\left(\1-\Lambda\right)\tau\fright]\tr\fleft[\left(\1-\Lambda\right)\rho\fright]^\frac{1-\alpha}{\alpha}\right\rvert^\alpha\right) \\
	&\geq\frac{1}{\alpha-1}\ln\left(\left\lvert\tr\fleft[\Lambda\tau\fright]\tr\fleft[\Lambda\rho\fright]^\frac{1-\alpha}{\alpha}\right\rvert^\alpha\right) \\
	&=\frac{1}{\alpha-1}\ln\left(\left\lvert\tr\fleft[\Lambda\tau\fright]\right\rvert^\alpha\tr\fleft[\Lambda\rho\fright]^{1-\alpha}\right).
\end{align}
Then the desired statement follows directly.
\end{proof}

\begin{proposition}
\label{prop:oneshot-state}
Let $\en{E}\equiv(p_{[r]},\rho_{[r]})$ be an ensemble of states with $p_{[r]}\in\itr(\s{P}_r)$ and $\rho_{[r]}\in\s{D}_A^{[r]}$.  Then for all $\alpha\in(1,+\infty)$,
\begin{align}
	&-\ln P_\abb{err}\fleft(\en{E}\fright) \notag\\
	&\quad\leq\sup_{s_{[r]}\in\s{P}_r}\inf_{\tau\in\aff\fleft(\s{D}_A\fright)}\sum_{x\in[r]}s_x\sw{D}_\alpha\fleft(\tau\middle\|\rho_x\fright)+\frac{\alpha}{\alpha-1}\ln\left(\frac{1}{p_{\min}}\right),
\end{align}
where $p_{\min}\equiv\min_{x\in[r]}p_x$.
\end{proposition}

\begin{proof}
Let $\tau\in\aff(\s{D}_A)$ be a unit-trace Hermitian operator, and let $\Lambda_{[r]}\in\s{M}_{A,r}$ be a POVM.  Applying Lemma~\ref{lem:converse-hoeffding}, we have that
\begin{align}
	1&=\sum_{x\in[r]}\tr\fleft[\Lambda_x\tau\fright] \\
	&\leq\sum_{x\in[r]}\left\lvert\tr\fleft[\Lambda_x\tau\fright]\right\rvert \\
	&\leq\sum_{x\in[r]}\left(\tr\fleft[\Lambda_x\rho_x\fright]\right)^\frac{\alpha-1}{\alpha}\exp\left(\frac{\alpha-1}{\alpha}\sw{D}_\alpha\fleft(\tau\middle\|\rho_x\fright)\right) \label{pf:oneshot-state-1}\\
	&=\sum_{x\in[r]}p_{\min}\left(\tr\fleft[\Lambda_x\rho_x\fright]\right)^\frac{\alpha-1}{\alpha} \notag\\
	&\quad\times\exp\left(\frac{\alpha-1}{\alpha}\sw{D}_\alpha\fleft(\tau\middle\|\rho_x\fright)+\ln\left(\frac{1}{p_{\min}}\right)\right) \\
	&\leq\sum_{x\in[r]}p_x\left(\tr\fleft[\Lambda_x\rho_x\fright]\right)^\frac{\alpha-1}{\alpha} \notag\\
	&\quad\times\exp\left(\frac{\alpha-1}{\alpha}\max_{x'\in[r]}\sw{D}_\alpha\fleft(\tau\middle\|\rho_{x'}\fright)+\ln\left(\frac{1}{p_{\min}}\right)\right) \\
	&\leq\left(\sum_{x\in[r]}p_x\tr\fleft[\Lambda_x\rho_x\fright]\right)^\frac{\alpha-1}{\alpha} \notag\\
	&\quad\times\exp\left(\frac{\alpha-1}{\alpha}\max_{x\in[r]}\sw{D}_\alpha\fleft(\tau\middle\|\rho_x\fright)+\ln\left(\frac{1}{p_{\min}}\right)\right). \label{pf:oneshot-state-2}
\end{align}
Here \eqref{pf:oneshot-state-1} follows from Lemma~\ref{lem:converse-hoeffding}; Eq.~\eqref{pf:oneshot-state-2} uses the concavity of the function $r\mapsto r^\frac{\alpha-1}{\alpha}$ for all $\alpha\in(1,+\infty)$.  Then we have that
\begin{align}
	-\ln\left(\sum_{x\in[r]}p_x\tr\fleft[\Lambda_x\rho_x\fright]\right)&\leq\max_{x\in[r]}\sw{D}_\alpha\fleft(\tau\middle\|\rho_x\fright)+\frac{\alpha}{\alpha-1}\ln\left(\frac{1}{p_{\min}}\right).
\end{align}
Since this holds for every unit-trace Hermitian operator $\tau$ and every POVM $\Lambda_{[r]}$, we conclude that
\begin{align}
	&-\ln P_\abb{err}\fleft(\en{E}\fright) \notag\\
	&\quad=\sup_{\Lambda_{[r]}\in\s{M}_{A,r}}-\ln\left(\sum_{x\in[r]}p_x\tr\fleft[\Lambda_x\rho_x\fright]\right) \\
	&\quad\leq\inf_{\tau\in\aff\fleft(\s{D}_A\fright)}\max_{x\in[r]}\sw{D}_\alpha\fleft(\tau\middle\|\rho_x\fright)+\frac{\alpha}{\alpha-1}\ln\left(\frac{1}{p_{\min}}\right) \\
	&\quad=\sup_{s_{[r]}\in\s{P}_r}\inf_{\tau\in\aff\fleft(\s{D}_A\fright)}\sum_{x\in[r]}s_x\sw{D}_\alpha\fleft(\tau\middle\|\rho_x\fright)+\frac{\alpha}{\alpha-1}\ln\left(\frac{1}{p_{\min}}\right). \label{pf:oneshot-state-3}
\end{align}
Here \eqref{pf:oneshot-state-3} follows from an application of the Sion minimax theorem~\cite{sion1958GeneralMinimaxTheorems}.
\end{proof}

Using the converse bound in Proposition~\ref{prop:oneshot-state} on the one-shot error probability, we strengthen Theorem~\ref{thm:asymptotic-state} and fully recover the log-Euclidean upper bound on $\overline{E}_\abb{err}(\en{E})$.

\begin{corollary}
\label{cor:asymptotic-state}
Let $\rho_{[r]}\in\s{D}_A^{[r]}$ be a tuple of states.  Then
\begin{align}
	\overline{E}_\abb{err}\fleft(\en{E}\fright)&\leq C^\flat\fleft(\rho_{[r]}\fright).
\end{align}
\end{corollary}

\begin{proof}
See our companion paper~\cite[Theorem~15]{ji2024BarycentricBoundsError}.
\end{proof}

\section{Quantum channel exclusion}
\label{sec:exclusion-channel}

In the task of quantum channel exclusion~\cite[Appendix~E~4]{huang2024ExactQuantumSensing}, the experimenter is faced with a processing device that implements an unknown channel.  The device is represented by an ensemble of channels, $\en{N}\equiv(p_{[r]},\ch{N}_{[r]})$ with $r\geq2$, $p_{[r]}\in\itr(\s{P}_r)$, and $\ch{N}_{[r]}\in\s{C}_{A\to B}^{[r]}$, and this indicates that for each $x\in[r]$, there is a prior probability $p_x\in(0,1)$ with which the device always implements the channel $\ch{N}_x\in\s{C}_{A\to B}$.  The experimenter's goal is to submit an index $x'\in[r]$ that \emph{differs} from the label of the channel that the device actually implements.

When the experimenter is allowed to invoke the processing device $n$ times, the most general strategy for channel exclusion, known as an \emph{adaptive} strategy, is represented by a quantum comb with $n$ empty slots~\cite{gutoski2007GeneralTheoryQuantum,chiribella2008MemoryEffectsQuantum,chiribella2008QuantumCircuitArchitecture} (see our companion paper~\cite[Section~V.A]{ji2024BarycentricBoundsError} for details).  For an adaptive strategy $\ch{S}^{(n)}\equiv(\rho,\ch{A}_{[n-1]},\Lambda_{[r]})$ with $n$ invocations, we let $\ch{Q}_x^{(n)}\equiv(\rho,\ch{A}_{[n-1]},\Lambda_x)$ denote the ``substrategy'' corresponding to the outcome $x\in[r]$, and we let $\zeta(\ch{Q}_x^{(n)};\ch{N})$ denote the probability of obtaining the outcome $x$ when invoking $n$ times the channel $\ch{N}$ with the strategy $\ch{S}^{(n)}$.  Then it is sensible to write $\ch{S}^{(n)}\equiv\ch{Q}_{[r]}^{(n)}$.  Consequently, the \emph{(nonasymptotic) error probability} of channel exclusion \emph{with $n$ invocations} for the ensemble $\en{N}$ is given by
\begin{align}
	P_\abb{err}\fleft(n;\en{N}\fright)&\coloneq\inf_{\ch{Q}_{[r]}^{(n)}}\sum_{x\in[r]}p_x\zeta\fleft(\ch{Q}_x^{(n)};\ch{N}_x\fright).
\end{align}
The \emph{(asymptotic) error exponents} of channel exclusion for the ensemble $\en{N}$ are defined as
\begin{align}
	\underline{E}_\abb{err}\fleft(\en{N}\fright)&\coloneq\liminf_{n\to+\infty}\sup_{\ch{Q}_{[r]}^{(n)}}-\frac{1}{n}\ln\left(\sum_{x\in[r]}p_x\zeta\fleft(\ch{Q}_x^{(n)};\ch{N}_x\fright)\right), \\
	\overline{E}_\abb{err}\fleft(\en{N}\fright)&\coloneq\limsup_{n\to+\infty}\sup_{\ch{Q}_{[r]}^{(n)}}-\frac{1}{n}\ln\left(\sum_{x\in[r]}p_x\zeta\fleft(\ch{Q}_x^{(n)};\ch{N}_x\fright)\right).
\end{align}

\subsection{Upper bound on the asymptotic error exponent}
\label{sec:asymptotic-channel}

We provide an alternative proof for the barycentric upper bound on the asymptotic error exponent $\underline{E}_\abb{err}(\en{N})$ of channel exclusion in our companion paper~\cite[Theorem~24]{ji2024BarycentricBoundsError}, based on the fact that the Belavkin--Staszewski channel divergence is a strong converse bound on the error exponent of asymmetric binary channel hypothesis testing~\cite[Theorem~49]{fang2021GeometricRenyiDivergence}.

\begin{lemma}[\!\!{\cite[Theorem~49]{fang2021GeometricRenyiDivergence}}]
\label{lem:converse-channel}
Let $\ch{T},\ch{N}\in\s{C}_{A\to B}$ be two channels, and let $(\ch{Q}_{[2]}^{(n)})_n$ be a sequence of two-outcome adaptive strategies.  If
\begin{align}
	\liminf_{n\to+\infty}-\frac{1}{n}\ln\zeta\fleft(\ch{Q}_1^{(n)};\ch{N}\fright)&>\g{D}\fleft(\ch{T}\middle\|\ch{N}\fright),
\end{align}
then
\begin{align}
	\limsup_{n\to+\infty}\zeta\fleft(\ch{Q}_1^{(n)};\ch{T}\fright)&=0.
\end{align}
\end{lemma}

\begin{theorem}
\label{thm:asymptotic-channel}
Let $\ch{N}_{[r]}\in\s{C}_{A\to B}^{[r]}$ be a tuple of channels.  Then
\begin{align}
	\underline{E}_\abb{err}\fleft(\en{N}\fright)&\leq\sup_{s_{[r]}\in\s{P}_r}\inf_{\ch{T}\in\s{C}_{A\to B}}\sum_{x\in[r]}s_x\g{D}\fleft(\ch{T}\middle\|\ch{N}_x\fright).
\end{align}
\end{theorem}

\begin{proof}
Let $\ch{T}\in\s{C}_{A\to B}$ be a state, and let $(\ch{Q}_{[r]}^{(n)})_n$ be a sequence of adaptive strategies such that  for every positive integer $n$, $\ch{Q}_{[r]}^{(n)}$ has $n$ invocations.  We assert that there exists $x_\star\in[r]$ such that
\begin{align}
	\label{pf:asymptotic-channel-1}
	\liminf_{n\to+\infty}-\frac{1}{n}\ln\zeta\fleft(\ch{Q}_{x_\star}^{(n)};\ch{N}_{x_\star}\fright)&\leq\g{D}\fleft(\ch{T}\middle\|\ch{N}_{x_\star}\fright).
\end{align}
To see this, consider the opposite situation where
\begin{align}
	\liminf_{n\to+\infty}-\frac{1}{n}\ln\zeta\fleft(\ch{Q}_x^{(n)};\ch{N}_x\fright)&>\g{D}\fleft(\ch{T}\middle\|\ch{N}_x\fright)\quad\forall x\in[r].
\end{align}
By Lemma~\ref{lem:converse-channel}, this implies that
\begin{align}
	\limsup_{n\to+\infty}\zeta\fleft(\ch{Q}_x^{(n)};\ch{T}\fright)&=0\quad\forall x\in[r],
\end{align}
which contradicts the fact that $\sum_{x\in[r]}\zeta\fleft(\ch{Q}_x^{(n)};\ch{T}\fright)=1$ for every positive integer $n$.  This shows by contradiction the existence of $x_\star\in[r]$ satisfying \eqref{pf:asymptotic-channel-1}.  It follows from \eqref{pf:asymptotic-channel-1} that
\begin{align}
	\min_{x\in[r]}\liminf_{n\to+\infty}-\frac{1}{n}\ln\zeta\fleft(\ch{Q}_{x_\star}^{(n)};\ch{N}_{x_\star}\fright)&\leq\max_{x\in[r]}\g{D}\fleft(\ch{T}\middle\|\ch{N}_{x_\star}\fright).
\end{align}
Since this holds for every channel $\ch{T}$ and every sequence of adaptive strategies $(\ch{Q}_{[r]}^{(n)})_n$, we infer that
\begin{align}
	&\sup_{\mleft(\ch{Q}_{[r]}^{(n)}\mright)_n}\min_{x\in[r]}\liminf_{n\to+\infty}-\frac{1}{n}\ln\zeta\fleft(\ch{Q}_x^{(n)};\ch{N}_x\fright) \notag\\
	&\quad\leq\inf_{\ch{T}\in\s{C}_{A\to B}}\max_{x\in[r]}\g{D}\fleft(\ch{T}\middle\|\ch{N}_x\fright). \label{pf:asymptotic-channel-2}
\end{align}
It follows that
\begin{align}
	\underline{E}_\abb{err}\fleft(\en{N}\fright)&=\liminf_{n\to+\infty}\sup_{\ch{Q}_{[r]}^{(n)}}-\frac{1}{n}\ln\left(\sum_{x\in[r]}p_x\zeta\fleft(\ch{Q}_x^{(n)};\ch{N}_x\fright)\right) \label{pf:asymptotic-channel-3}\\
	&=\sup_{\mleft(\ch{Q}_{[r]}^{(n)}\mright)_n}\liminf_{n\to+\infty}-\frac{1}{n}\ln\left(\sum_{x\in[r]}p_x\zeta\fleft(\ch{Q}_x^{(n)};\ch{N}_x\fright)\right) \label{pf:asymptotic-channel-4}\\
	&=\sup_{\mleft(\ch{Q}_{[r]}^{(n)}\mright)_n}\liminf_{n\to+\infty}\min_{x\in[r]}-\frac{1}{n}\ln\left(\zeta\fleft(\ch{Q}_x^{(n)};\ch{N}_x\fright)\right) \\
	&\leq\sup_{\mleft(\ch{Q}_{[r]}^{(n)}\mright)_n}\min_{x\in[r]}\liminf_{n\to+\infty}-\frac{1}{n}\ln\left(\zeta\fleft(\ch{Q}_x^{(n)};\ch{N}_x\fright)\right) \\
	&\leq\inf_{\ch{T}\in\s{C}_{A\to B}}\max_{x\in[r]}\g{D}\fleft(\ch{T}\middle\|\ch{N}_x\fright) \label{pf:asymptotic-channel-5}\\
	&=\sup_{s_{[r]}\in\s{P}_r}\inf_{\ch{T}\in\s{C}_{A\to B}}\sum_{x\in[r]}s_x\g{D}\fleft(\ch{T}\middle\|\ch{N}_x\fright). \label{pf:asymptotic-channel-6}
\end{align}
Here \eqref{pf:asymptotic-channel-4} follows because the supremum on the right-hand side of \eqref{pf:asymptotic-channel-3} can be replaced with a maximum; Eq.~\eqref{pf:asymptotic-channel-5} follows from \eqref{pf:asymptotic-channel-2}; Eq.~\eqref{pf:asymptotic-channel-6} follows from Ref.~\cite[Lemma~3]{ji2024BarycentricBoundsError}.
\end{proof}

\section*{Acknowledgements}
\label{sec:acknowledgements}

KJ is grateful to Bartosz Regula for bringing Ref.~\cite[Exercise~3.57]{hayashi2017QuantumInformationTheory} to his attention.  The formation of this work benefited from the conference ``Beyond IID in Information Theory,'' held at the University of Illinois Urbana-Champaign from July 29 to August 2, 2024, and supported by NSF Grant No.~2409823.  KJ acknowledges support from the NSF under grant no.~2329662.  HKM and MMW acknowledge support from the NSF under grant no.~2304816 and grant no.~2329662 and AFRL under agreement no.~FA8750-23-2-0031.  The work of MM was partially funded by the National Research, Development and Innovation Office of Hungary via the research grants K 146380 and EXCELLENCE 151342, and by the Ministry of Culture and Innovation and the National Research, Development and Innovation Office within the Quantum Information National Laboratory of Hungary (Grant No.~2022-2.1.1-NL-2022-00004).

\begin{figure*}[t]
\centering
\subfloat[The strong converse part of the quantum Stein's lemma. \label{fig:converse}]{\includegraphics[scale=0.18]{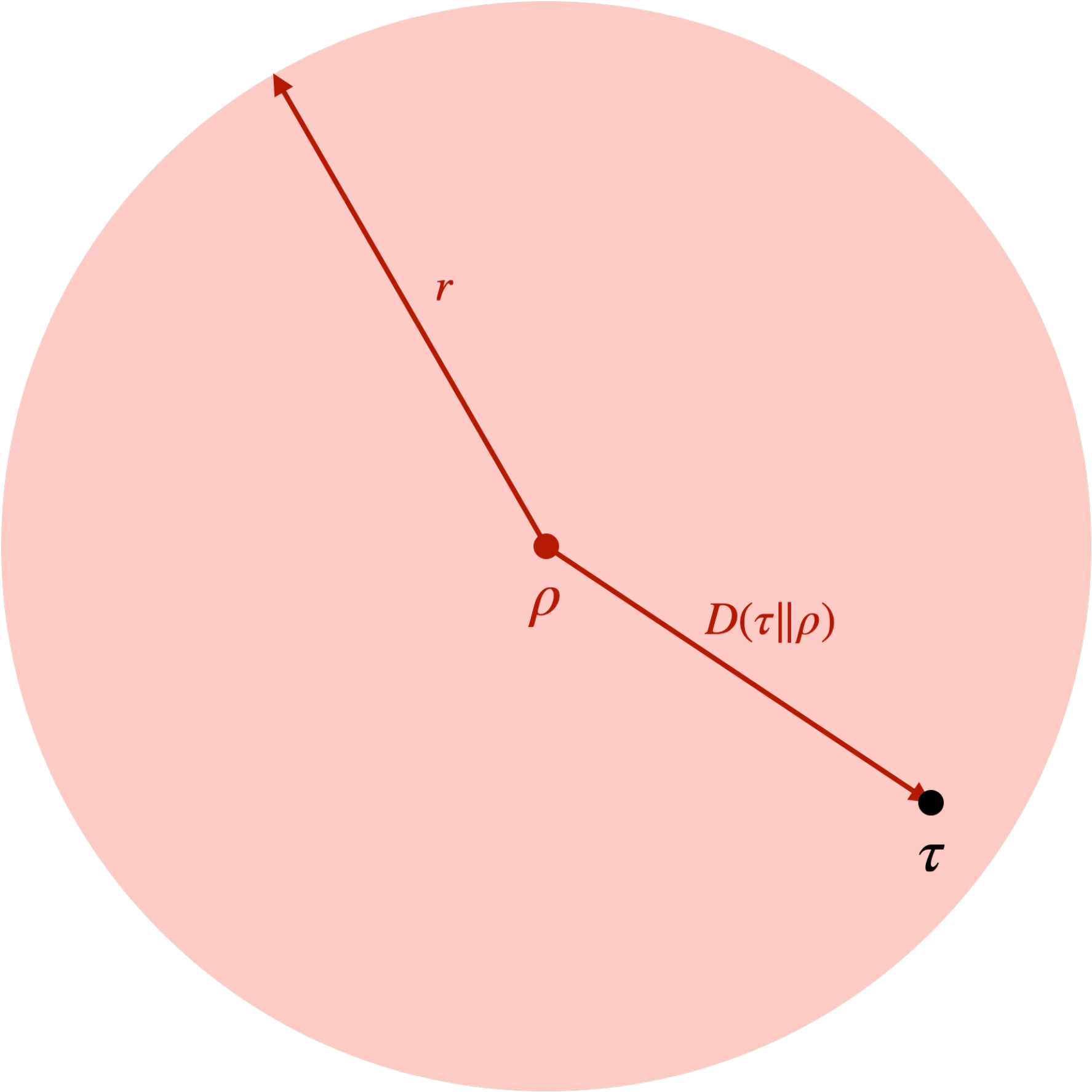}} \qquad\qquad
\subfloat[A contradiction led by a ``dummy'' state if the error exponent exceeds the divergence radius. \label{fig:contradiction}]{\includegraphics[scale=0.18]{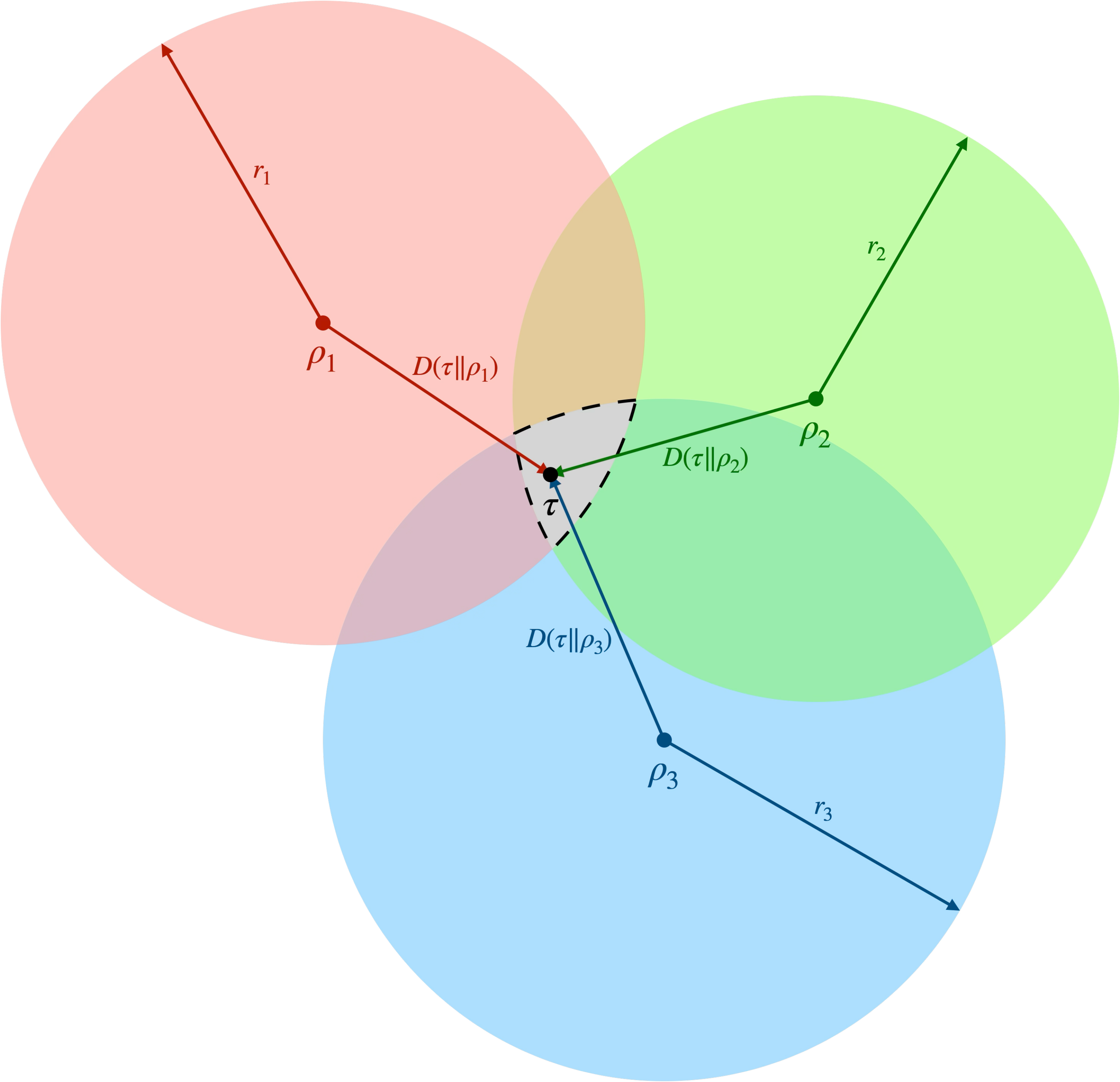}} \\
\subfloat[The error exponent cannot exceed the divergence radius. \label{fig:radius}]{\includegraphics[scale=0.18]{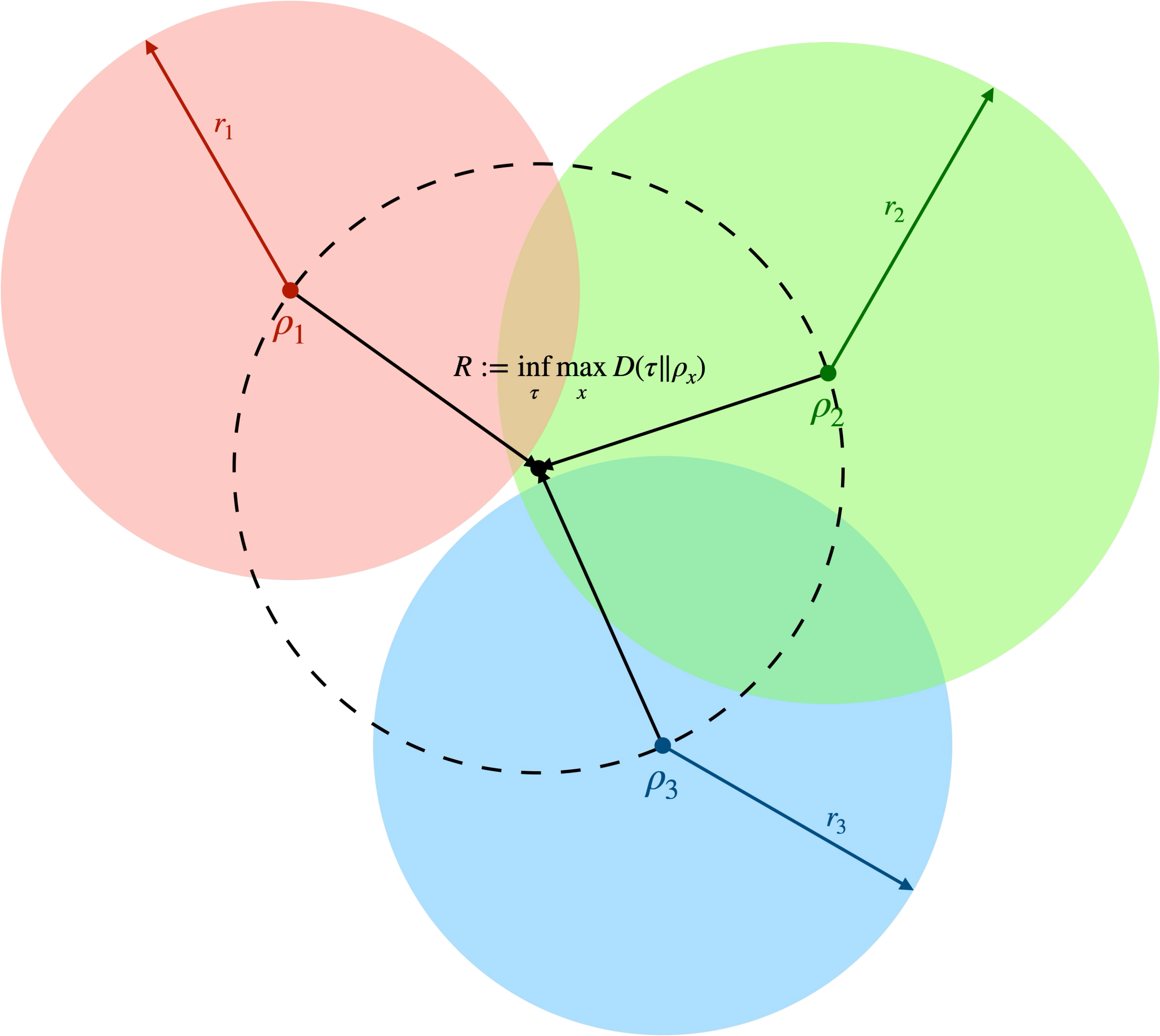}}
\caption{(a) The strong converse part of the quantum Stein's lemma can be understood in a pictorial way as follows.  Let $(\Lambda^{(n)})_n$ be a sequence of measurement operators (i.e., with $0\leq\Lambda^{(n)}\leq\1^{(n)}$ for every positive integer $n$) such that $-\frac{1}{n}\ln\tr[\Lambda^{(n)}\rho^{\otimes n}]\to r$ as $n\to+\infty$.  Then any state $\tau$ inside the $r$-sphere around $\rho$, in the sense that $D(\tau\|\rho)<r$, must satisfy $\tr[\Lambda^{(n)}\tau^{\otimes n}]\to 0$. \quad (b) Consider a sequence of POVMs $(\Lambda_{[3]}^{(n)})_n$ such that $-\frac{1}{n}\ln\tr[\Lambda_x^{(n)}\rho_x^{\otimes n}]\to r_x$ for $x=1,2,3$.  Then consider a state $\tau$ in the intersection of the $r_1$-sphere around $\rho_1$ (red), the $r_2$-sphere around $\rho_2$ (green), and the $r_3$-sphere around $\rho_3$ (blue).  This implies that $\tr[\Lambda_x^{(n)}\tau^{\otimes n}]\to 0$ for $x=1,2,3$, which contradicts $\sum_{x}\tr[\Lambda_x^{(n)}\tau^{\otimes n}]=1$ for all $n$. \quad (c) Therefore, the intersection of the three spheres must be an empty set.  This implies that $-\frac{1}{n}\ln\sum_{x}\tr[\Lambda_x^{(n)}\rho_x^{\otimes n}]\to\min\{r_1,r_2,r_3\}\leq R\coloneq\inf_{\tau}\max_{x}D(\tau\|\rho_x)$.  Since this holds for every sequence of POVMs $(\Lambda_{[3]}^{(n)})_n$, the divergence radius $R$ provides an upper bound on the asymptotic error exponent of quantum state exclusion.}
\label{fig}
\end{figure*}


\clearpage
\bibliographystyle{IEEEtran}
\bibliography{Library}

\end{document}